\documentclass[12pt]{amsart}
\usepackage{amsmath,amsfonts,amssymb,geometry}
\usepackage{graphicx,amsmath, amscd}
\usepackage{color}

\def\be{\begin{equation}}
\def\ee{\end{equation}}
\def\bea{\begin{eqnarray}}
\def\eea{\end{eqnarray}}
\def\bt{\begin{theorem}}
\def\et{\end{theorem}}
\def\bl{\begin{lemma}}
\def\el{\end{lemma}}
\def\br{\begin{remark}}
\def\er{\end{remark}}
\def\bc{\begin{corollary}}
\def\ec{\end{corollary}}
\def\bd{\begin{definition}}
\def\ed{\end{definition}}
\def\eprop{\end{proposition}}
\def\bprop{\begin{proposition}}
\def\a{\alpha}
\def\b{\beta}
\def\l{\lambda}

\def\e{\varepsilon}

\def\r{\rho}
\def\s{\sigma}

\def\cD{\mathcal{D}}
\def\cE{\mathcal{E}}

\def\cS{\mathcal{S}}

\def\cU{\mathcal{U}}

\def\bC{\mathbf{C}}

\def\bR{\mathbf{R}}

\def\bbC{\mathbf{C}}

\def\bbI{\mathbb{I}}

\def\b1{B_{1}^}

\def\ba{\begin{array}}
\def\ea{\end{array}}
\def\ben{\begin{enumerate}}
\def\een{\end{enumerate}}
\def\iy{\infty}
\renewcommand{\L}{\Lambda}
\newtheorem{lemma}{Lemma}[section]
\newtheorem{theorem}[lemma]{Theorem}
\newtheorem{proposition}[lemma]{Proposition}
\newtheorem{corollary}[lemma]{Corollary}

\newtheorem{definition}[lemma]{Definition}

\newcommand{\Id}{\mathbb{Id}}

\newcommand{\K}{\mathcal{APPT}_{\textnormal{sym}}}

\renewcommand{\leq}{\leqslant}
\renewcommand{\geq}{\geqslant}

\DeclareMathOperator{\vol}{\mathrm{vol}}

\DeclareMathOperator{\tr}{\mathrm{tr}}

\newcommand{\R}{\mathbf{R}}
\newcommand{\C}{\mathbf{C}}

\DeclareMathOperator{\E}{\mathbb{E}}
\renewcommand{\P}{\mathbb{P}}

\newcommand{\SEP}{\mathcal{APPT}_0}

\newcommand{\st}{\  : \ }

\begin{document}
\title[The APPT property for random induced states]{
The absolute positive partial transpose property\\ for random
induced states} 
\author{Benoit Collins, Ion Nechita and Deping Ye}
\date{}
\begin{abstract}
In this paper, we first obtain an algebraic formula for the moments
of a centered Wishart matrix, and apply it to obtain new convergence results
in the large dimension limit when both parameters of the distribution tend to infinity at
different speeds.

We use this result to investigate APPT (absolute positive partial
transpose) quantum states. We show that the threshold for a
bipartite random induced state on $\C^d=\C^{d_1} \otimes
\C^{d_2}$, obtained by partial tracing a random pure state on
$\C^d \otimes \C^s$, being APPT occurs if the environmental
dimension $s$ is of order $s_0=\min(d_1, d_2)^3 \max(d_1, d_2)$.
That is, when $s \geq Cs_0$, such a random induced state is
APPT with large probability, while such a random states is not
APPT with large probability when $s \leq cs_0 $. Besides, we
compute effectively $C$ and $c$ and show that it is possible to
replace them by the same sharp transition constant when
$\min(d_1, d_2)^{2}\ll d$.

\end{abstract}
\maketitle

\section{Introduction}

Geometry of quantum states strives to understand the
geometric properties of subsets of quantum states, and has
attracted considerable attention, especially in the case of the
large dimension \cite{AS1, AubrunSzarekWerner2011, fhs, Sz1,
SzarekWerner2008,SzarekWerner2011,Deping2009, Deping2010}. The
high dimensional setting is common (and of particular interest) in
Quantum Information Theory whose building blocks, quantum states,
often are objects with huge dimension (for instance, the set of
quantum states on the system $(\bC^3)^{\otimes 8}$ has dimension
43046720). This high dimensional setting indicates the importance
of random constructions which now is a main tool in understanding
the typical behavior of random induced states. To generate random
induced states, one often relies on random matrices. The
connections between Random Matrices and Quantum Information Theory
were pushed forward by Hayden, Leung, and Winter in their studies
of aspects of generic entanglement \cite{Hayden2006}. Together
with tools from Geometric Functional Analysis and Convex
Geometric Analysis, random matrices and random constructions have
led to many important (and even unexpected) results, such as Hastings's disproof of the famous additivity conjecture for the classical capacity of quantum channels \cite{Hasting2009}. Recent
contributions include the studies of the generic properties for
entanglement vs. separability,  and PPT (positive partial
transpose) vs non-PPT \cite{Aubrun2011, AubrunSzarekYe2011,
AubrunSzarekYe2011a, Nechita2011}.

Detecting quantum entanglement, a phenomenon first discovered
in \cite{EPR1} and now being the key ingredient of quantum
algorithms (see \cite{Nie1, Sh1}), is one of the fundamental
problems in Quantum Information Theory. Among those necessary
and/or sufficient conditions for separability and entanglement,
the Peres-Horodecki PPT criterion \cite{Ho1, Pe1}
 is the simplest but the most powerful one. The Peres-Horodecki
 PPT criterion is a necessary condition and is sufficient only for
 the systems $\bC^2\otimes
 \bC^2$ and $\bC^2\otimes \bC^3$ \cite{St1, Wo1}. From the
  computational complexity point of view,
 separability and PPT are quite different: determining separability
 is an NP-hard problem \cite{Gurvits-NPHard}, but determining PPT is easy since
 it only requires
 to verify the eigenvalues of the partial transpose of given states being positive.
 Note that both the separability and PPT are encoded in the spectral properties
 of quantum states in a complicate way. Necessary and/or sufficient conditions on
 determining separability and PPT by just the information of eigenvalues
 (referred to as the absolute separability and absolute PPT in
 literature) could be very useful in Quantum Information Theory
 as it can help to reduce the cost of
 storage spaces and (processing) time. For absolute separability,
 less results are known;
 however, necessary and sufficient conditions for APPT have been found by Hildebrand
 \cite{Hildebrand2007}.
Understanding when a random induced states is APPT is the main motivation of this work.

To that end, we first prove that a properly centered
$d\times d$ Wishart matrix of parameter $s$ has its expected
normalized moments that can be written as a polynomial in the
variables $d$ and $d/s$. The coefficients of this polynomial have
a simple combinatorial interpretation, and some families of
coefficients are known. This algebraic fact has important
consequences in the two parameter asymptotic study of the Wishart
matrix. Indeed, it allows in particular to capture the precise
nature of the behavior of the Wishart matrix in the case where
$d\to\infty$, and in particular in the case where $s/d\to \infty$
too. We summarize our first main result as follows:

\vskip 2mm
\noindent {\bf Theorem A.} {\em Let $W_{d}$ be a
$d\times d$ Wishart matrix of parameter $s$ and let
$Z_{d}=\sqrt{ds}\left(\frac{W_{d}}{ds}-\frac{\Id}{d}\right)$ be
its centered and renormalized version. The moments of $Z_{d}$ are
given by
\begin{equation*}
\mathbb E  \frac{1}{d}\mathrm{tr}\left[Z_{d}\right]^p =
\!\!\!\!\! \sum_{\substack{\alpha\in S_{p}\\\alpha \text{ has no fixed points}}} \!\!\!\!\! d^{-2g(\alpha )}\left(\frac{d}{s}\right)^{|\alpha|-p/2}.
\end{equation*}
Moreover, almost surely as $d\to\infty$ and $s/d\to \infty$,  the extremal eigenvalues of $Z_d$ converge to $\pm 2$.}

Both results (for moments and for extremal eigenvalues) are of
separate interest in random matrix theory, where the single
scaling $d\to \infty$, $s/d \to c >0$ has received a lot of
attention\cite{elkaroui,johnstone,nechita}. The above result is
then applied to estimate the threshold for a random induced
state being APPT vs. non-APPT. We have the following result (we
put $p=\min(d_1,d_2)$):

\vskip 2mm \noindent {\bf Theorem B.} {\em There are effectively
computable absolute constants $c, C>0$, such that,  if $\rho$ is a
bipartite random induced state on $\C^d=\C^{d_1} \otimes
\C^{d_2}$, obtained by partial tracing over $\C^s$ a random
pure state on $\bC^d \otimes \C^s$, then for $d=d_1d_2$ large
enough, one has:
\begin{itemize}
\item[(i)] The quantum state $\rho$ is APPT with very large probability when $s \geq (4+\varepsilon)p^2 d$.
\item[(ii)] The quantum state $\rho$ is {\em not APPT} with very large probability when $s \leq cp^2 d$. If $p^{2}\ll d$, one can take $c=4 - \varepsilon$; when $p^{2}$ is of the order of $d$, $c$ can be computed as described in the proof of Theorem \ref{thm:threshold-p-infty}.
\end{itemize}
 }

The letters $C,c,c_0,...$ denote absolute numerical constants
(independent of anything) whose value may change from place to
place. When $A,B$ are quantities depending on the dimension (and
perhaps some other parameters), the notation $A \lesssim B$ means
that there exists an absolute constant $C>0$ such that the
inequality $A \leq CB$ holds in every dimension. Similarly $A
\simeq B$ means both $A \lesssim B$ and $B \lesssim A$. As usual,
$A\sim B$ means that  $A/B\to 1$ as the dimension (or some other
relevant parameter) tends to $\infty$, while $A = o(B)$ means that
$A/B\to 0$. For a $d \times d$ complex matrix $A \in M_d(\C)$ we denote by $\mathrm{tr}(A)$ its non-normalized trace. In this paper, whenever we deal with a tensor product structure $\C^d = \C^{d_1} \otimes \C^{d_2}$, we put $p=\min(d_1,d_2)$.

The article is organized as follows. Section
\ref{sec:combinatorics-Wishart} is of random matrix theoretic
flavor. Using mostly combinatorics and also a little bit of
analysis and elementary probability we obtain new formulas for the
moments of centered Wishart matrices and estimates on their
extremal eigenvalues. Section \ref{sec:properties-APPT} gathers
some properties about the the APPT property and section
\ref{sec:bounds-APPT} provides the bounds of the threshold
for APPT.

\section{Combinatorics of centered Wishart matrices}\label{sec:combinatorics-Wishart}
In this section, we prove two results, Theorems \ref{thm:moments-centered-Wishart} and \ref{thm:convergence-extreme} about centered and renormalized Wishart matrix. These results are interesting for random matrix theorists and can be considered independently from the rest of the paper.

\label{Combinatorics:Wishart}
\subsection{Preliminaries and notation}
We start by introducing some notation from combinatorics. For an
integer $p$, we denote $[p]=\{1, 2, \ldots, p\}$, and
$[0]=\emptyset$. For a subset $I \subset [p]$, let $\mathcal S_I$
be the set of \emph{permutations} which act on $I$. We shall take
the convention that $\mathcal S_\emptyset=\{\emptyset\}$. The set of permutations without fixed points will de denoted by
$$\mathcal S^o_I = \{\alpha \in \mathcal S_I \, | \, \forall i \in I, \, \alpha(i) \neq i\}.$$
The \emph{length} $|\alpha|$ of a permutation $\alpha \in \mathcal
S_I$ is the minimal number of transpositions needed to decompose
$\alpha$. We put $|\emptyset| = 0$. The notation
$|\cdot|$ is polymorphic, since it is used to denote both the
cardinality of sets and the length of permutations. We shall also
use the notation $\# \alpha$ for the number of cycles of $\alpha$;
the following relation holds $\# \alpha + |\alpha| = |I|$.

For a nonempty subset $I$, we denote by $\gamma_I$ the \emph{full
cycle} in $\mathcal S_I$, with elements of $I$ ordered
increasingly: $\gamma_I = (i_1 \, i_2 \, \cdots \, i_k)$, where
$I=\{i_1 < i_2 < \cdots < i_k\}$. Abusing notation, we define
$|\gamma_\emptyset|=-1$. In this way, the geodesic inequality
$$|\alpha|+|\alpha^{-1}\gamma_I| \geq |\gamma_I|=|I|-1$$
holds for all $I$ and all $\alpha \in \mathcal S_I$. We define the
\emph{genus} of a permutation $\alpha$ as half of the amount by
which the above inequality fails to be an equality
$$g_I(\alpha) = \frac{|\alpha|+|\alpha^{-1}\gamma_I| -
|I|+1}{2}.$$ It is a standard fact in combinatorics that the genus
$g_I(\alpha)$ is a nonnegative integer. For $\alpha \in \mathcal
S_I$, define $\tilde \alpha \in \mathcal S_J$ to be $\alpha$
without its fixed points; in other words, $J=\{i \in I \, | \,
\alpha(i) \neq i\}$ and, for $j \in J$, we have $\tilde \alpha(j)
= \alpha(j)$. It is clear that $|\tilde \alpha|= |\alpha|$.
Moreover, by the following lemma, erasing fixed points leaves the
genus of the permutation unchanged.
\begin{lemma}\label{lem:genus-permutaion-fixed-points} Let $\tilde
\alpha \in \mathcal S_J$ be the permutation $\alpha \in \mathcal
S_I$ with its fixed points removed. Then, $g_J(\tilde \alpha) =
g_I(\alpha)$.
\end{lemma}

\begin{proof} Going in the opposite direction and
proceeding by induction, it suffices to show that whenever we add
a fixed point $i$ to a permutation $\tilde \alpha \in \mathcal
S_J$, its genus remains unchanged. Let us denote by $j_1$ and
$j_2$ the neighboring points in $J$ between which $i$ is inserted:
$j_1 < i < j_2$. Also, we note by $j_1-1$ the predecessor of $j_1$
in $J$ and by $j_2+1$ the successor of $j_2$ in $J$. Using the
number of cycles notation, one needs to show that $\#(\tilde
\alpha^{-1}\gamma_J) = \#(\alpha^{-1}\gamma_I)$.

Given two permutations $\sigma, \tau \in \mathcal S_p$, recall the following  combinatorial interpretation of the number of cycles of $\sigma^{-1}\tau$. Define a multigraph $G_{\sigma, \tau} = (V,E)$ with $2p$ vertices $V=\{1,  \ldots, p, 1', \ldots, p'\}$ and edges
$$E = \{(k, \sigma(k)') \, , \, k \in [p]\} \cup \{(k, \tau(k)') \, , \, k \in [p]\}.$$
Then the number of cycles $\#(\sigma^{-1}\tau)$ equals the number of connected components of $G_{\sigma, \tau}$.

Going back to our setting, it is clear that the graphs $G_{\tilde \alpha, \gamma_J}$ and $G_{\alpha, \gamma_I}$ have the same number of connected components, since adding the extra vertices $i, i'$ does not alter the edge structure of $G_{\tilde \alpha, \gamma_J}$; for a sketch of the argument, see Figure \ref{fig:genus-cycles}.
\end{proof}
\begin{figure}[htbp]
\centering
\includegraphics{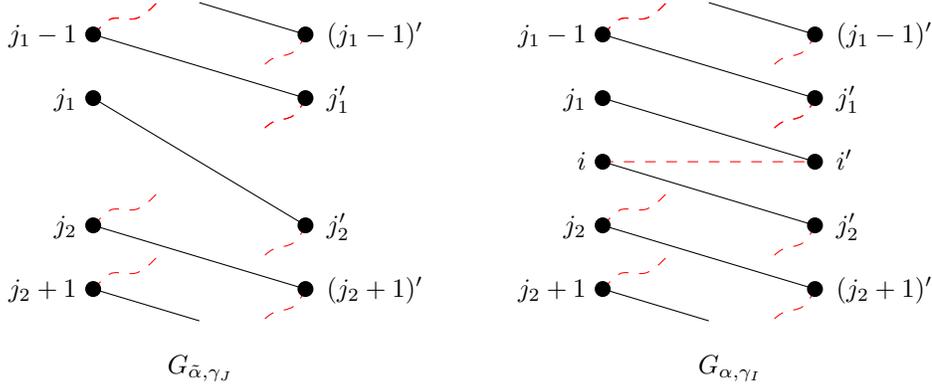}
\caption{Adding a fixed point $i$ to a permutation $\tilde \alpha$ does not increase the number of cycles of $\tilde \alpha^{-1}\gamma_J$. Edges corresponding to $\alpha$ and $\tilde \alpha$ are represented by dashed red lines and edges corresponding to full cycles are black and solid.}
\label{fig:genus-cycles}
\end{figure}

\subsection{A moment formula for the centered Wishart matrix}
\label{sec:moments-centered-Wishart}

Let $G \in M_{d \times s}(\mathbf C)$ be a Ginibre random matrix
(i.e. $\{G_{ij}\}$ are i.i.d. standard complex Gaussian
random variables) and $W= W_{d} = GG^*$ be the
corresponding Wishart matrix of parameters $(d,s)$, where $G^*$ denotes the Hermitian adjoint of $G$. Here we make an abuse of notation and keep track only of the parameter $d$,
considering implicitly that $s$ will be a function of $d$. It is
easy to see that $\mathbb E W_{ij} = s \delta_{ij}$ and $\mathbb E
\mathrm{tr} W = ds$.

The main theorem of this section characterizes the fluctuations of
$W$ around its mean. For this purpose, we introduce the
following $d\times d$ matrix, a properly rescaled, centered
Wishart matrix of parameters $(d,s)$:
$$Z_{d}=\sqrt{ds}\left(\frac{W_{d}}{ds}-\frac{\Id _d}{d}\right),$$
where $\Id _d$ refers to the $d\times d$ identity matrix. We show
that the following theorem holds true.

\begin{theorem}\label{thm:moments-centered-Wishart}
The moments of $Z_d$ are given by
\begin{equation}\label{eq:moments-centered-Wishart}
\mathbb E  \frac{1}{d}\mathrm{tr}\left[Z_{d}\right]^p = \!\!\!\!\!
\sum_{\alpha\in S^o_{p}}  d^{-2g(\alpha
)}\left(\frac{d}{s}\right)^{|\alpha|-p/2}.
\end{equation}
\end{theorem}

Note that despite the simplicity of this combinatorial result, it seems to be new.
Before we prove this result, we would like to describe three corollaries, obtained by letting one or both parameters $d$ and $s$ go to infinity.

\begin{corollary}
If $d$ is fixed and $s \to \infty$, then
$$\lim_{s \to \infty} \mathbb E  \frac{1}{d}\mathrm{tr}\left[Z_{d}\right]^p =
\begin{cases}
0, & \qquad \text{if $p$ is odd},\\
\sum_{g=0}^\infty \varepsilon(p/2,g)d^{-2g}, & \qquad \text{if $p$ is even},\\
\end{cases}$$
where $\varepsilon(p/2,g)$ is the number of products of $p/2$
disjoint transpositions in $\mathcal S_p$ of genus $g$.
Alternatively, for even $p$, $\varepsilon(p/2,g)$ is known to
count the number of gluings of a $p$-gon into a surface of genus
$g$.
\end{corollary}

\begin{proof}
When $s \to \infty$, $d/s \to 0$ the only
terms in equation \eqref{eq:moments-centered-Wishart} which
survive are those for which $|\alpha|=p/2$. It follows that
$\alpha$ must be in this case a product of $p/2$ disjoint
transpositions (for even $p$). Reordering the sum by genera gives
the statement (see \cite{zvo}).
\end{proof}

Note that the sequence $\varepsilon(p,g)$ appears in the Online
Encyclopedia of Integer Sequences \cite{oeis} as A035309.

For the forthcoming corollary, we need to recall that the set
$NC(p)$ is the collection of partitions of $[p]$ that have no
crossings with respect to the canonical order. Moreover, we
introduce the subset
$$NC^o(p) = \{ \pi \in NC(p) \, | \, \pi \text{ has no singletons}\}.$$

\begin{corollary}
If both $d,s\to \infty$ such that $s/d\to c$ for some
constant $c>0$, we obtain
$$\lim_{\substack{d,s \to \infty\\ s/d \to c}} \mathbb E  \frac{1}{d}\mathrm{tr}\left[Z_{d}\right]^p =
\sum_{\pi\in NC^o(p)} c^{\#\pi-p/2},$$
where $\# \pi$ denotes the number of blocks of the partition $\pi$. In particular, the random matrix $Z_{d}$ converges in moments to a \emph{centered} Marchenko-Pastur distribution of parameter $c$ (rescaled by $\sqrt c$).
\end{corollary}

\begin{proof} In this asymptotic regime, the surviving
terms in equation \eqref{eq:moments-centered-Wishart} are those
for which $g(\alpha)=0$. The formula in the statement follows from
a well known result of Biane \cite{bia} saying that the permutations in $\mathcal S_p$
of genus $0$ are in one to one correspondence with non-crossing
partitions $\pi \in NC(p)$. The second part follows from a
centered version of the free Poisson limit theorem \cite[Theorem
12.11]{nsp}. For $c=1$, the rescaled quantities appearing in the
statement are the Riordan numbers $R_p$ \cite[sequence A005043
]{oeis} such that $R_p = |NC^o(p)|$; see \cite{npe} for
the connection between Riordan numbers and centered free Poisson
random variables.
\end{proof}

Finally, we have the following general asymptotics

\begin{corollary}\label{cor:Zd-semicircular}
If $d \to \infty$ and $s/d \to \infty$ (in other words $1 \ll d \ll s$), then
$$\lim_{\substack{d\to \infty\\ s/d \to \infty}} \mathbb E  \frac{1}{d}\mathrm{tr}\left[Z_{d}\right]^p =
\begin{cases}
0, & \qquad \text{if $p$ is odd},\\
\mathrm{Cat}_{p/2}, & \qquad \text{if $p$ is even},\\
\end{cases}$$
where $\mathrm{Cat}_n$ is the $n$-th Catalan number. In particular,
the random matrix $Z_{d}$ converges in moments
to a standard semicircular distribution.
\end{corollary}

\begin{proof}
This follows directly from the fact that
the Catalan numbers count the number of non-crossing pairings of
$[2p]$: $\varepsilon(2p,0) = \mathrm{Cat}_p$.
\end{proof}

Note that the above result also follows from the more general
result by \cite{bayi}, where the almost sure convergence is also
obtained. A proof of the almost sure convergence could
also be obtained in the combinatorial spirit of this paper, for
instance along the lines of \cite{coneaap}.

We would like to explain briefly why this result is not surprising
from a heuristic point of view. Indeed, the distribution for
Marcenko-Pastur distribution, as given in Proposition
\ref{prop:RMT-Wishart}, is
$$\pi_c=\frac{\sqrt{4c-(x-1-c)^2}}{2\pi x}\mathbf{1}_{[(\sqrt c - 1)^2,(\sqrt c +1)^2]}(x)\,dx,$$
so as $c\to\infty$,
 the distribution of $Z_{d}$ should approach
$$\sqrt{s/d}\frac{\sqrt{(2+\sqrt{d/s}-x)(x-2+\sqrt{d/s})}}{2\pi (x+\sqrt{s/d})}
\mathbf{1}_{[-2+\sqrt{d/s},2+\sqrt{d/s}]}(x)\,dx$$
which should tend to the semi circle distribution.
Our corollary therefore implies that we let $d,s/d$ go to infinity separately or together with any correlation we like.

Let us now prove the combinatorial result.

\begin{proof}[Proof of Theorem \ref{thm:moments-centered-Wishart}.]
The starting point is a
formula for the moments of $W$, obtained via the Wick calculus
(for an intuitive graphical approach to this problem, see
\cite{coneaap}):
$$\mathbb E \mathrm{tr}(W^p) = \sum_{\alpha \in \mathcal S_p}
d^{\#(\alpha^{-1}\gamma)}s^{\#\alpha} = (ds)^p\sum_{\alpha \in
\mathcal S_p}d^{-|\alpha^{-1}\gamma|}s^{-|\alpha|},$$ where
$\gamma \in \mathcal S_p$ is the forward cycle
$\gamma=(1\,2\,\cdots \,p)$. By applying the binomial formula for
the commuting matrices $W$ and $\Id _d$, we get
\begin{align*}
m_p&:=\mathbb E  \frac{1}{d}\mathrm{tr}\left[\sqrt{ds}
\left(\frac{W}{ds}-\frac{\Id _d}{d}\right)\right]^p \\
& = d^{-1+p/2}s^{p/2}\sum_{I \subset [p]}(-1)^{|I^c|}(ds)^{|I|}
d^{-|I^c|}\mathbb E \mathrm{tr}(W^{|I|})\\
& = \sum_{I \subset [p]}\sum_{\alpha \in \mathcal S_I}
(-1)^{p-|I|}d^{-1-p/2+|I|-|\alpha^{-1}\gamma|}s^{p/2-|\alpha|}\\
& = \sum_{I \subset [p]}\sum_{\alpha \in \mathcal S_I}(-1)^{p-|I|}
d^{-2g_I(\alpha)}\left(\frac{d}{s}\right)^{|\alpha|-p/2}.
\end{align*}
To conclude, we need to show that the terms in the sum above
cancel out, except for the ones with $I=[p]$ and $\alpha \in
\mathcal S_{[p]}$ without fixed points. For a permutation $\tilde
\alpha \in \mathcal S^o_J$, denote
by $[\tilde \alpha]$ the set of permutations which extend $\tilde
\alpha$ by adding fixed points:
$$[\tilde \alpha] = \{\alpha \in \mathcal S_I \, | \, J
\subset I \text{ , } \alpha(j) = \tilde \alpha(j) \quad \forall
j \in J \text{ and } \alpha(i) = i \quad \forall i \in I \setminus J\}.$$
Regrouping terms in the sum and using the fact that for $\alpha
\in [\tilde \alpha]$, $|\alpha|=|\tilde \alpha|$ and
$g_I(\alpha)=g_J(\tilde \alpha)$, we can write
$$m_p = \sum_{\tilde \alpha \in \mathcal S^o}d^{-2g_J(\tilde \alpha)}\left(\frac{d}{s}\right)^{|\tilde \alpha|-p/2}\sum_{\alpha \in \mathcal S_I \cap [\tilde \alpha]}(-1)^{p-|I|},$$
where the first sum in the equation above is indexed by
permutations $\tilde \alpha$ without fixed points. Given such a
permutation $\tilde \alpha \in \mathcal S_J$, for every larger set
$I \supset J$ there is a unique way of extending $\tilde \alpha$
to $\alpha \in \mathcal S_I$. Hence, the second sum in the above
equation is given by
$$\sum_{J \subset I \subset [p]}(-1)^{p-|I|} = \delta_{J, [p]},$$
which can be understood as a M\"obius inversion formula in the poset
of subsets of $[p]$. In conclusion, only the permutations $\tilde \alpha
\in \mathcal S^o_p$ give non-zero contribution.
To finalize the proof, note that such a permutation has at most $p/2$ cycles
and thus at least length $p/2$.
\end{proof}

\subsection{Almost sure convergence of extremal eigenvalues of the
Wishart matrices in the regime $d\ll s$}

We start with the following lemma:
\begin{lemma}\label{moment-bound}
Let us assume that $d/s\to c\in (0,\infty )$ as $d\to\infty$. For
any $\varepsilon >0$ there exists a constant $C>0$ such that for
any $p\leq \sqrt{d}$, we have
$$\E\mathrm{tr} Z_{d}^{p} \leq C(2+\sqrt{c}+\varepsilon)^{p}.$$
\end{lemma}

\begin{proof} Under the hypotheses of the lemma, the
eigenvalues counting measure of $Z_{d}$ converges almost surely to
a probability measure whose support is
$[-2+\sqrt{c},2+\sqrt{c}]$.

Let us first evaluate $\P(\|Z_{d}\| \geq 2+\sqrt{c})$,
where $\|Z_d\|$ is the operator norm of $Z_d$. According to the
union bound, this is bounded above by $$\P(\lambda_1(Z_d)\geq 2+\sqrt{c})+ \P(\lambda_d (Z_d)\leq
-2-\sqrt{c}),$$ where $\lambda_1(Z_d)$ and
$\lambda_d(Z_d)$ are respectively the largest and smallest eigenvalues of
$Z_d$.

According to Theorem 3 of Soshnikov in \cite{sos}, there exists a constant $C>0$ such that for any $p\leq \sqrt{d}$,
we have
$$\E\mathrm{tr} W_{d}^{p} \leq C(\sqrt{c}+1)^{2p}.$$
By Jensen inequality this implies under the same assumption on $p$ that
$$\P(\lambda_{1}\geq (\sqrt{c}+1)^{2}+t)\leq \frac{(\sqrt{c}+1)^{2p}}{ [(\sqrt{c}+1)^{2}+t]^{p}}.$$

Besides, it follows from Equation (15) in \cite{kaca}, that
the probability of having eigenvalues less than
$t<(\sqrt{c}-1)^{2}$ is less than $\exp (-d g(t))$. For our
purposes it is enough to know that $g(t)>0$ as long as
$0<t<(\sqrt{c}-1)^{2}$.

We can conclude the proof of the lemma from the two above observations
via the inequalities
$$\lambda_{1}^{p}\leq \mathrm{tr} W^{p}\leq n\lambda_{1}^{p},$$
and the formula
$$\E(X)=\int_{0}^{\infty}\P(X\geq t)dt.$$
\end{proof}

Now we would like to let $s/d \to \infty$ simultaneously (but independently) with $d \to \infty$. This is the setting of our next result, namely:
\begin{theorem}\label{thm:convergence-extreme}
Almost surely, when $1 \ll d \ll s$, the extremal eigenvalues of $Z_{d}$ converge to $\pm 2$.
\end{theorem}

\begin{proof}
Let $\varepsilon >0$ and $d_{0}$ large
enough such that for $d\geq d_{0}$, $\sqrt{d/s}\leq \varepsilon
/2$. The moment inequality from Lemma \ref{moment-bound} together
with the fact that our moment formula of Theorem
\ref{thm:moments-centered-Wishart} involves only positive terms
implies that there exists a constant $C>0$, such that for all
$p\leq \sqrt{d}$,
$$\E(\mathrm{tr} Z_{d}^{p}) \leq C(2+\varepsilon)^{p}.$$

Letting $d\to\infty$, by the Borel-Cantelli Lemma and the same
Jensen inequality as in the proof of Lemma \ref{moment-bound},
 we obtain that $\limsup \lambda_{1}\leq 2+\varepsilon$
almost surely. Since this holds true for all $\varepsilon >0$ and
since $\liminf \lambda_1 \geq 2$ by Corollary \ref{cor:Zd-semicircular}, we
get the desired result.
\end{proof}

An interesting aspect of the above proof is that it relies on
moment techniques, and therefore makes use of Theorem
\ref{thm:moments-centered-Wishart}. Here the interest of the
moment method and combinatorics is that they explain why one can
make rigorous a change of limit between $d\to\infty$ and
$s/d\to\infty$ regarding the almost sure convergence of the
largest eigenvalue.

We could have obtained directly this result with complex analysis
results (see e.g. \cite{boga}) and one could probably have
obtained more refined estimates (e.g. large deviation bounds,
universality results, etc); however, this was not in the spirit of
our combinatorial approach and we leave it for future investigation.

Using similar techniques, one could extend the above results to
show the following quantitative bound: for all $\varepsilon >0$,
with exponential small probability in $d \to \infty$, the spectrum
of a random density matrix $\rho$ from the induced ensemble of
parameters $(d,s)$ is contained in the interval
$$\left[\frac{1}{d}-\frac{2(1+\varepsilon)}{\sqrt{ds}},\frac{1}{d}+\frac{2(1+\varepsilon)}{\sqrt{ds}}\right].$$

\section{Existence of a threshold for $\mathcal{APPT}$}\label{sec:properties-APPT}

\subsection{Quantum states and Absolute PPT}
We now introduce some necessary notation and concepts related to
quantum information theory; readers are referred to \cite{Ben1}
and \cite{Nie1} for more details.

Consider a (complex) Hilbert space $\bC^d=\bC^{d_1} \otimes
\cdots\otimes \bC^{d_k}$ with (complex) dimension $d=d_1\cdots
d_k$, where $d_i\geq 2$ for all $i=1, \cdots, k$. The set of
states on $\bC^d$ (denoted by $\cD=\cD(\bC^d)$) can be identified
with the set of $d\times d$ density matrices. That is,
$$\cD=\{\r \in M^\text{sa}_d(\C) \, | \, \r \geq 0 \text{ and } \mathrm{tr}\r = 1\}.$$
Clearly the real
dimension of $\cD$ is $N=d^2-1$. The set $\cD$ is contained in the
affine hyperplane
$$H_1=\{A\in M^\text{sa}_d(\C) \, | \, \mathrm{tr} A = 1\},$$
endowed with the Hilbert-Schmidt inner product $\langle A,B\rangle = \tr (A^*B)$.

Partial tracing states on $\bC^d\otimes \bC^s$ over $\bC^s$ gives
reduced density matrices of size $d\times d$. Any state $\r$ on $
\bC^{d}\otimes\bC^{s}$ may be written as
$$\r = \sum_{i,j}^{d} \sum _{\a, \beta }^{s} \r_{i\a, j\beta }E_{i\alpha, j\beta},$$
where $E_{i\alpha, j\beta}$ are the matrix units associated to orthonormal bases $\{e_i\}_{i=1}^{d}$ and
$\{f_\a\}_{\a=1}^{s}$ of $\bC ^{d}$ and $\bC ^{s}$ respectively. The partial trace of $\r$ over
$\bC^s$, denoted by $\s=\mathrm{tr}_{\bC^s}(\r)$ may be
formulated as
$$ \s_{ij}=\sum _{\beta=1}^{s}
\r_{i\beta, j\beta} \ \mathrm{for}\ i,j=1, \cdots, d.$$

The \emph{induced measure} on $\cD(\C^{d})$ by partial tracing
over $\bC^s$ is an important probability distribution and can be
described as follows. Let $|\psi \rangle \langle \psi |$ be a
random pure state on $\C^d \otimes \C^s$, where $\psi$ is a random
unit vector uniformly distributed on the sphere
 in $\C^d \otimes \C^s$. Then the random induced state
 $\r=\mathrm{tr}_{\bC^s}(|\psi \rangle \langle
 \psi |)\in \cD(\bC^d)$ follows the distribution $\mu_{d,s}$.
Equivalently, one can find a $d\times s$ matrix $M$ distributed
uniformly on the sphere of $d\times s$ matrices, such that
$\r=MM^{*}.$ The distribution $\mu_{d,s}$ plays central roles
in this section. When $s=d$, one gets $\mu_{d,d}$, the normalized
Lebesgue (i.e., the Hilbert-Schmidt) measure on $\cD(\bC^{d})$.
Hence, a random state distributed according to $\mu_{d,d}$ is
uniformly distributed on $\cD(\C^d)$. When $s \geq d$, the
probability measure $\mu_{d,s}$ has a simple form \cite{Zycz2}
\begin{equation} \label{eq:formula-density} \frac{d\mu_{d,s}}{d\! \vol}(\rho)
= \frac{1}{Z_{d,s}} (\det \rho)^{s-d} ,\end{equation} where
$Z_{d,s}$ is a normalization factor. Note that formula
\eqref{eq:formula-density} allows to define the measure
$\mu_{d,s}$ (in particular) for every real $s \geq d$, while the
partial trace construction makes sense only for integer values of
$s$.

Hereafter, we will focus on the bipartite system, i.e, the case
$\bC^d=\bC^{d_1}\otimes \bC^{d_2}$ with $d=d_1d_2$. The partial
transpose operator (denoted by $T_2$) is a linear operation
that consists in taking the transpose in one leg and doing nothing
on the other leg, i.e., $T_2(\tau_1\otimes \tau_2)=\tau_1\otimes
\tau_2^{T}$, where $T$ is the normal transpose operator. The set
of quantum states with positive partial transpose is denoted by
$\mathcal{PPT}$, i.e., $\r\in \mathcal{PPT}$ if and only if
$T_2(\r)\geq 0$. Geometrically, $\mathcal{PPT}=\cD\cap
T_2(\cD)$, and $\mathcal{PPT}$ is a convex body with constant
height \cite{SBZ}. Peres-Horodecki PPT criterion states that
$\cS\subset \mathcal{PPT}$ \cite{Ho1, Pe1}, and
$\cS=\mathcal{PPT}$ only if $d\leq 6$ \cite{St1, Wo1}. Here, the
$N=d^2-1$ dimensional set $\cS \subset \cD$ is the set of
separable quantum states on $\bC^d$ \cite{Wer1} defined as
$$\cS =\cS(\bC^{d_1}\otimes \bC^{d_2}):=
{\rm conv}\{\rho _1 \otimes \rho _2, \ \ \rho _1\in
\cD(\bC^{d_1}),\ \ \r_2\in \cD (\bC ^{d_2})\}.$$ (Similarly, one
can define $\cS(\bC^{d_1}\otimes \cdots\otimes \bC^{d_k})$.) The
set $\cE:=\cD\setminus \cS$ is the set of entangled quantum
states, which play a crucial role in quantum information and
quantum computing.

 A quantum state $\r\in \cD(\bC^{d_1}\otimes \bC^{d_2})$ is \emph{absolutely
PPT}(or APPT) if for any unitary matrix $U\in \cU(d)$, $U\r U^*\in
\mathcal{PPT}$. The set of all states being APPT, denoted
as $\mathcal{APPT}$, is
\begin{equation} \mathcal{APPT}=\bigcap_{U\in \cU(d)}U
(\mathcal{PPT})U^* \subset \mathcal{PPT}.\label{APPT:def:1}
\end{equation} Clearly, $\mathcal{APPT}$ is a convex body, a
convex compact set with non-empty interior. This follows easily
from (\ref{APPT:def:1}) and the following result in
\cite{BarnumGurvits2002}: $\epsilon
\cD+(1-\epsilon)\frac{\bbI}{d}\subset \cS \subset \mathcal{PPT}$
for some $\epsilon<\frac{1}{d-1}$.

In applications, one often requires the convex bodies to be
 origin-symmetric. In this section, we will mainly work on the symmetric
convex body $\K = - \SEP \cap \SEP$ where
$\SEP=\mathcal{APPT}-\mathbb{I}/d$. Such a symmetrization will not
change many geometric parameters of interest (such as, the volume
radius and mean width) substantially, due to the famous
Rogers-Shephard inequality \cite{Rogers1957}. Both $\SEP$
and $\K$ sit in the {\em linear} hyperplane
$$H_0=\{A\in M^\text{sa}_d(\C) \, | \, \mathrm{tr} A = 0\}.$$

Recall that the gauge associated to a convex body $K\subset \bR^N$ is the
function $\|\cdot\|_K$ defined by
\[ \|x\|_K = \inf \{ t \geq 0 \st x \in tK \}, \ \ \ \forall x
\in \R^N.\] Note that $\|x\|_K = \|-x\|_K$ for origin-symmetric
convex bodies $K$. The outradius of $K$ is the smallest $R>0$ such
that $K$ is contained in a ball of radius $R$. Similarly, the
inradius of a convex body $K$ is the largest radius $r$ of a
Euclidean ball contained in $ K$. For origin-symmetric convex
bodies, $r$ and $R$ can be defined as the ``best'' constants such
that $R^{-1}|\cdot| \leq \|\cdot\|_K \leq r^{-1}|\cdot|$ (where
$|\cdot|$ denotes the Euclidean norm.) Let $X$ be a standard
Gaussian vector in $\R^N$, i.e., a random vector with independent
$N(0,1)$ coordinates in any orthonormal basis.

\begin{proposition}
In the notation of the present section, we have
\begin{equation}\label{bounds}
d^{1/2} | \cdot | \leq \|\cdot\|_{\K} \leq d |\cdot| \iff d^{-1} | \cdot | \leq \|\cdot\|_{\K^\circ} \leq d^{-1/2} |\cdot|,
\end{equation}
\begin{equation}\label{sym-upper}
 \E \|X\|_{\SEP} \leq \E \|X\|_{\K} \leq 2\E \|X\|_{\SEP}.
\end{equation}
\end{proposition}

\begin{proof} Any matrix $A \in \SEP$ satisfies $A \geq
-\Id/d$. This implies that any $A \in \K$ satisfies $-\Id/d \leq A
\leq \Id/d$, or $\|A\|_{\iy} \leq 1/d$, and therefore the
outradius of $\K$ is bounded by $1/\sqrt{d}$. The inradius of
$\mathcal{APPT}$ equal to $(d(d-1))^{-1/2}$ follows directly from
Corollary 3 in \cite{BarnumGurvits2002}. This proves
(\ref{bounds}).

Note that the distribution of $X$ is symmetric, and
\begin{align*}
\|X\|_{\SEP} &\leq \|X\|_{\K} \\
&= \max (\|X\|_{\SEP},\|-X\|_{\SEP})
\leq \|X\|_{\SEP} + \|-X\|_{\SEP}.
\end{align*}
Then (\ref{sym-upper}) follows after taking expectation.
\end{proof}

It was pointed out that the set $\mathcal{APPT}$ varies if the
decomposition of $d$ changes: if $\min(d_1', d_2')\geq \min(d_1,
d_2)$, then any APPT quantum states on $\bbC^{d_1'}\otimes
\bbC^{d_2'}$ must be APPT on $\bbC^{d_1}\otimes \bbC^{d_2}$
\cite{Hildebrand2007}. Consequently, the largest APPT set is obtained
when $d_1=2$ and $d_2=d/2$, and the smallest APPT set is obtained when
$d_1=d_2=\sqrt{d}$.

Let $p=\min(d_1, d_2)$. Denote by $S_+=\{(k, l): 1\leq k\leq l\leq
p\}$ and $S_{-}=\{(k,l): 1\leq k<l\leq p\}$. Note that the
cardinalities of the sets $S_+$ and $S_{-}$ are $p_+=p(p+1)/2$ and
$p_{-}=p(p-1)/2$ respectively. Let $$\s_{+}: S_+\rightarrow \{1,
\cdots, p_+\}, \ \ \mathrm{and} \ \ \s_{-}: S_{-}\rightarrow \{1,
\cdots, p_{-}\},$$ be two orderings (i.e. bijective maps) on $S_+$
and $S_{-}$ respectively. Thus, $\s_{+}(k, l)\leq p_+$ and
$\s_{-}(k, l)\leq p_{-}$ for all pairs $(k,l)$.

For $\l=(\l_1, \cdots, \l_d)$ with ordering $\l_1\geq \l_2\geq
\cdots \geq \l_d\geq 0$ and $\sum_{i=1}^{d}\l_i=1$, one defines
the matrix $\L(\l; \s_{+}, \s_{-})$ as
$$\big(\L(\l; \s_+, \s_{-})\big)_{k,l}=\left\{
\begin{array}{cc}
\l _{d+1-\s_+(k,l)}, & k\leq l;\\
-\l _{\s_{-}{(l,k)}},& k>l.\end{array}\right.$$

Define the $p\times p$ symmetric matrix $\Theta(\l; \s_{+},
\s_{-})$ to be the sum of $\L(\l; \s_+, \s_{-})$ and its transpose
$\L^{T}(\l; \s_+, \s_{-}).$ Note that $\Theta(\l; \s_+, \s_{-})$
has the following form:
\begin{equation}\label{Theta:definition}\Theta_{k,l}=
\left\{\begin{array}{cc} 2\l_{a(k)}, & k=l,\\
\l_{b(k,l)}-\l_{c(k,l)}, & b(k,l)>c(k,l), k>l,
\end{array}\right.\end{equation} where $a(k)$, $b(k, l)$ and $c(k,l)$ are
some integer-valued functions with values smaller than or equal to
$d$. Thus $\Theta_{k, l} \leq 0$ for all $k\neq l$.

The following theorem is a necessary and sufficient condition for
$\r\in \mathcal{APPT}$ (see Theorem III.9 or Lemma III.10 in
\cite{Hildebrand2007} ).

\begin{theorem}\label{APPT:equivalence:condition}
Let $\r$ be a quantum state on $\bC^d=\bbC^{d_1}\otimes\bbC^{d_2}$
with $d=d_1d_2$. Suppose that $\r$ has eigenvalues $\l_1\geq
\l_2\geq \cdots\geq \l_d\geq 0$. Then $\r$ is $\mathcal{APPT}$ if
and only if for all pair of ordering $(\s_+, \s_{-})$, $\Theta(\l;
\s_+, \s_{-})$ is positive semi-definite.
\end{theorem}

\subsection{Threshold for $\mathcal{APPT}$ is $\sim w(\SEP ^\circ)^2$.}

Hereafter, $w(\SEP ^\circ)$ denotes the mean width of $\SEP
^\circ$, the polar body of $\SEP$. For general convex body $K$
with the origin in its interior, its polar body (denoted by
$K^\circ$) can be defined as
\[ K^\circ = \{ y \in \R^N \st \langle x,y \rangle \leq 1 \ \ \forall x \in K
\},\] where $\langle \cdot ,\cdot \rangle$ denotes the usual inner
product and induces the Euclidean norm $|\cdot|$. The mean width
of $K$, $w (K)$, is defined as
\begin{equation}\label{mean:width:definition} w (K):=\int _{S^{N-1}}
h_K(u)\,d\s(u)=\int _{S^{N-1}} \|u\|_{K^\circ} d\s(u),
\end{equation} where $\,d\s(u)$ is the normalized spherical
measure on the sphere $S^{N-1}$, and $h_K(u)=\max _{x\in K}
\langle x, u\rangle = \|u\|_{K^\circ}$ for any  $u\in S^{N-1}$. A
more convenient quantity to calculate is the Gaussian mean width
of $K$
\begin{equation} \label{defn-mean-width} w_G(K)= \E \|X\|_{K^\circ} = \E \sup_{x \in K}
\langle X,x \rangle,\end{equation} where $X$ is a standard
Gaussian vector in $\R^N$. By passing to polar coordinates, one
can easily check that for every convex body $K \subset \R^N$
\begin{equation} \label{gammaN}
w_G(K) = \gamma_N \, w(K),
\end{equation}
where \be \label{gammaNbounds}
 \gamma_N = \E |G| = \frac{\sqrt{2}\Gamma((N+1)/2)}{\Gamma(N/2)},  \quad
\sqrt{N-1}\leq \gamma_N \leq \sqrt{N}, \ee
 is a constant depending only on $N$.
We set $s_0(d_1, d_2)$ to be
$$s_0=s_0(d_1, d_2)=\left(\frac{\E\|X\|_{\SEP}}{d^2}\right)^2=
\left(\frac{w_{G}(\SEP ^\circ)}{d^2}\right)^2\sim w(\SEP
^\circ)^2.$$ By inequality (\ref{sym-upper}), one has $s_0\sim
w(\SEP ^\circ)^2\simeq w(\K ^\circ)^2$. The following theorem
states that the threshold for the set $\mathcal{APPT}$ is of order
of $w(\SEP ^\circ)^2$.

\bt\label{Threshold:APPT:existence}  There are effectively
computable absolute constants $C,c>0$, such that, if $\rho$ is a
random induced state on $\C^{d_1} \otimes \C^{d_2}$ distributed
according to the measure $\mu_{d,s}$, then
\begin{itemize}
 \item [(i)] $\P(\rho \in \mathcal{APPT}) \leq C
 \exp(-cs)$ for $s \leq cs_0$;
 \item [(ii)] $\P(\rho \notin \mathcal{APPT})
 \leq C \exp(-cs_0)$ for $s \geq Cs_0$.
\end{itemize}  \et

We first point out that the threshold value for the set
$\mathcal{PPT}$ occurs only at those $s\geq 2d$. For the balanced
bipartite case (i.e. $d_1=d_2$) it follows from Theorem 4 in
\cite{Aubrun2011}, while for the unbalanced bipartite case (i.e.,
$d_1\neq d_2$) it follows from \cite{Nechita2011}. As
$\mathcal{APPT}\subset \mathcal{PPT}$, the threshold for
$\mathcal{APPT}$ must be larger and thus we also have $s\geq 2d$.

The following lemma is our main tool to prove that the threshold
for the set $\mathcal{APPT}$ can be taken as $w(\SEP^\circ)^2$.
This lemma aims to approximate $\rho-\Id/d$ by
$\frac{1}{d\sqrt{s}} X$ where $X$ is a standard Gaussian vector in
the space $H_0$ of traceless Hermitian $d \times d$ matrices. We
refer readers to its detailed proof in \cite{AubrunSzarekYe2011}.

\begin{lemma} \label{lem:MP-vs-GUE}
For every convex body $K \subset H_0$ containing $0$ in its
interior, and for every $s \geq d$, if $\rho$ is a random state on
$\C^d$ distributed according to $\mu_{d,s}$, and if $X$ is a
standard Gaussian vector in $H_0$, then
\[ \E \left\| \rho- \frac{\Id}{d} \right\|_K \simeq \frac{1}{d\sqrt{s}} \E \| X\|_{K} .\]
\end{lemma}

 Applying the lemma for
$K=\SEP$, we obtain that \be \label{MPorder} \E \left\| \rho-
\frac{\Id}{d} \right\|_{\SEP} \simeq \sqrt{\frac{s_0(d)}{s}} . \ee
This suggests that the threshold for the set $\SEP$ occurs at
$s_0(d)$, since a state $\rho$ is APPT when $\|\rho-\Id/d\|_{\SEP}
\leq 1$ and non-APPT when $\|\rho-\Id/d\|_{\SEP} > 1$.

The proof of Theorem \ref{Threshold:APPT:existence} is almost
identical to that of Section 4 in \cite{AubrunSzarekYe2011}, and
here we sketch its proof for completeness. We refer the readers to
\cite{AubrunSzarekYe2011} for more details, in particular Appendix
E for the L\'{e}vy's Lemma and concentration of measure theory.

\begin{proof}[Proof of Theorem \ref{Threshold:APPT:existence}.]
Let $S_{HS}$ be the Hilbert-Schmidt sphere in the space of $d
\times s$ matrices (it can be identified with the real sphere
$S^{2ds-1}$) and $f: S_{HS} \to \R$ be the function defined by \[
f(M) = \left\| MM^* - \frac{\Id}{d} \right\|_{\SEP} .\]
Formula (\ref{MPorder}) asserts that $\E f\simeq
\sqrt{\frac{s_0(d)}{s}}.$ For every $r>0$, denote by
$\Omega=\Omega (r)$ the subset
\[ \Omega = \{ M \in S_{HS} \st \|M\|_{\iy} \leq r \} .\]  Inequality \eqref{bounds}
and the proof of Lemma 4.2 in \cite{AubrunSzarekYe2011} imply that
the Lipschitz constant of $f|_{\Omega}$ is bounded by $2rd$. Note
that $\Omega (r)=S_{HS}$ since $\|M\|_{\iy}\leq \|M\|_2=1$. Then,
the global Lipschitz constant of $f$ is bounded by $2d$, and hence
the median of $f$ (denoted $M_f$) differs from its mean, $\E f$,
by at most $C2d/\sqrt{2ds}=C'\sqrt{d/s}$ (see Appendix E in
\cite{AubrunSzarekYe2011}). It follows that the median of $f$  is
also of order $\sqrt{s_0/s}$.

By a net argument similar to that in
\cite{AubrunSzarekWerner2011}, one has $\P(S_{HS} \setminus
\Omega) \lesssim \exp (-cs)$  if $r=3/\sqrt{d}$. A local version
of L\'evy's lemma (see Appendix E in \cite{AubrunSzarekYe2011})
implies that for $\e=M_f/2 \simeq \sqrt{s_0/s}$,
\begin{align*}
\P( |f-M_f| \geq M_f/2)&=\P( |f-M_f| \geq \e) \\
&\lesssim \P(S_{HS} \setminus \Omega) + \exp(-c_1ns(\e/2dr)^2) \\
& \lesssim \exp(-cs) +\exp(-cs_0).
\end{align*}
Therefore, one has
$$\P(\rho \textnormal{ is APPT}) = \P( f \leq 1) \lesssim \exp(-cs),$$
whenever $M_f\geq 2$ (or, equivalently, $s \lesssim
s_0$) and
$$\P(\rho \textnormal{ is not APPT}) = \P( f > 1) \lesssim
\exp(-cs_0),$$
whenever $M_f \leq 2/3$ (or, equivalently, $s \gtrsim s_0$).
 This ends the proof of Theorem \ref{Threshold:APPT:existence}.
\end{proof}

\section{Estimates on the threshold for $\mathcal{APPT}$}\label{sec:bounds-APPT}
We consider the product systems $\mathbf C^d = \mathbf C^{d_1}
\otimes \mathbf C^{d_2}$. Recall that a random state $\rho$ on
$\C^d= \C^{d_1}\otimes \C^{d_2}$ distributed according to
$\mu_{d,s}$ has the same distribution as $MM^*$, where $M$ is a $d
\times s$ matrix uniformly distributed on the Hilbert-Schmidt
sphere of the Hilbert space of $d \times s$ complex matrices. A
more convenient, but equivalent, way is to link the measure
$\mu_{d,s}$ with a normalized Wishart matrix. More precisely, let
$W  = GG^*$, where $G \in M_{d\times s}(\mathbf C)$ is a Ginibre
matrix, i.e. a matrix with i.i.d. standard complex Gaussian
entries. The random matrix $W$ is called a Wishart matrix of
parameters $(d,s)$. A \emph{random induced state} $\r$ distributed
according to $\mu_{d,s}$ is then given by the formula $\rho = W /
\mathrm{tr} W$ \cite{nechita, Zycz2}.

The following theorem is the main result of this section, providing asymptotic values of the threshold for the convex body $\mathcal{APPT}$. This theorem covers the case when $p = \min(d_1, d_2) \to \infty$. The simpler case of $p$ bounded is treated in Theorem \ref{thm:threshold-p-finite} as it gives sharper bounds.

\begin{theorem}\label{thm:threshold-p-infty}
Let $\rho$ be a random induced state distributed according to
the probability measure $\mu_{d, s}$.
\begin{itemize}
\item[(i)] For all $\varepsilon
>0$, almost surely, when $d \to \infty$ and
$s>(4+\varepsilon)p^2d$, the quantum state $\rho$ is APPT;
\item[(ii)] When $1 \ll p^2 \ll d$ and $s<(4-\varepsilon)p^2d$, $\rho$
is not APPT almost surely;
\item[(iii)] When $p^2 \sim \tau d$ for a constant $\tau \in (0,1]$,
there exists a constant $C_\tau$ (see formula \eqref{eq:Ctau})
such that whenever and $s<4(C_\tau-\varepsilon)p^2d$, $\rho$ is
not APPT almost surely.
\end{itemize}
\end{theorem}

\begin{proof} We start with (i). For given eigenvalues
$\{\l_1, \cdots, \l_d\}$, we introduce the following $p\times p$
matrix:
$$\Upsilon=\Upsilon(\lambda;\sigma_+, \sigma_-)=
\Theta(\lambda;\sigma_+, \sigma_-) -2d^{-1} \Id _p,$$ where $\Id
_p$ denotes the $p\times p$ identity matrix. Recall that
$p=\min(d_1, d_2)$. From formula \eqref{Theta:definition} and
Theorem \ref{thm:convergence-extreme}, the matrix $\Upsilon$ has
small entries: $|\Upsilon _{ij}|\leq (4+\varepsilon)/\sqrt{ds}$. A
necessary condition for the matrix $\Theta = 2d^{-1} \Id_p +
\Upsilon$ to be semidefinite positive is that $\Upsilon$ should
have operator norm smaller than $2/d$. It is a well known fact in
matrix analysis (see \cite{hj}) that
$$ \| \Upsilon \|_\text{op} \leq p \| \Upsilon \|_{1 \to \infty} = p \max_{1 \leq i,j \leq p} | \Upsilon_{ij} | \leq \frac{(4+\varepsilon)p}{\sqrt{ds}}.$$
The conclusion in the statement follows by asking that
$(4+\varepsilon)p/{\sqrt{ds}} \leq 2/d$.

We move now to the proofs of (ii) and (iii). We shall proceed by
exhibiting a vector $x \in \mathbf R^p$, such that, $x^T
\Lambda(\lambda;\sigma_+,\sigma_-)x<0$ for some pair of linear
orderings. This does indeed suffice to show that the matrix
$\Theta$ is not semidefinite positive. Indeed, we take the column
vector $x=(1,1, \ldots, 1)^T \in \mathbf R^p$. Any pair of linear
orderings is compatible with such a vector, and one has
\begin{equation}\label{eq:test-with-ones}
x' \Lambda(\lambda;\sigma_+,\sigma_-)x = \sum_{i,j=1}^p
\Lambda_{ij} = \sum_{i=1}^{p_+} \lambda_{d+1-i} -
\sum_{i=1}^{p_{-}} \lambda_{i},
\end{equation} where $p_+=p(p+1)/2$ and $p_{-}=p(p-1)/2$.

We shall now consider the two regimes in the statement, starting
with $1 \ll p^2 \ll d$. The main idea here is to note that, for all
$i$, $\lambda_i = 1/d + \tilde \lambda_i / \sqrt{ds}$, where
$\tilde \lambda_i$ are the eigenvalues of the matrix $Z_d$
introduced in Section \ref{sec:moments-centered-Wishart}.
By Theorem \ref{thm:convergence-extreme}, for all $\varepsilon>0$ and for
$d$ large enough, all the ``large'' eigenvalues $\lambda_1,
\ldots, \lambda_{p_-}$ appearing in equation
\eqref{eq:test-with-ones} are bigger than $1/d +
(2-\varepsilon)/\sqrt{ds}$; in the same vein, all the ``small''
eigenvalues $\lambda_{d+1-p_+}, \ldots, \lambda_d$ are smaller
than $1/d - (2-\varepsilon)/\sqrt{ds}$.
We obtain
$$x' \Lambda(\lambda;\sigma_+,\sigma_-)x
\leq p_+\left(\frac{1}{d} - \frac{2-\varepsilon}{\sqrt{ds}}
\right) - p_-\left(\frac{1}{d} + \frac{2-\varepsilon}{\sqrt{ds}}
\right) = \frac{p}{d} - \frac{(2-\varepsilon)p^2}{\sqrt{ds}},$$
which is negative as long as $s<(4-\varepsilon)p^2d$. This
concludes the proof for the first regime $1\ll p^2 \ll d$.

We now move on the the second regime, where $p^2 \sim \tau d$, for
a fixed constant $\tau \in (0,1]$. Writing equation (\ref{eq:test-with-ones}) in terms of the
eigenvalues of $Z_d$, we obtain
$$x' \Lambda(\lambda;\sigma_+,\sigma_-)x =
\frac{p}{d} + \frac{1}{\sqrt{ds}}\left[\sum_{i=1}^{p_+} \tilde
\lambda_{d+1-i} -  \sum_{i=1}^{p_{-}} \tilde\lambda_{i}\right].$$
Using the fact that the asymptotic spectrum $\{\tilde \lambda_i\}$
of $Z_d$ is semicircular, we obtain the following bounds for $d$
large enough:
\begin{equation*}
\frac{1}{d}\sum_{i=1}^{p_+} \tilde \lambda_{d+1-i} \lesssim
\int_{-2}^{-c_{\tau/2}-\varepsilon} x w(x)\mathrm{d}x, \ \ \
\mathrm{and} \ \ \ \frac{1}{d}\sum_{i=1}^{p_{-}} \tilde\lambda_{i}
\gtrsim \int_{c_{\tau/2}+\varepsilon}^2 x w(x)\mathrm{d}x,
\end{equation*}
where
$$w(x) = \frac{1}{2\pi}\sqrt{4-x^2}$$
is the density of the standard semicircular distribution and
$c_{\tau/2} \in [0,2]$ is defined implicitly by
$$\int_{c_{\tau/2}}^2 w(x)\mathrm{d}x = \tau/2.$$

Indeed, it is a classical result in random matrix theory that the
conclusions of Theorem \ref{thm:convergence-extreme} and of
Corollary \ref{cor:Zd-semicircular} for the matrix model $Z_{d}$
imply that its repartition function converges almost surely
uniformly towards the repartition function of the semi-circle
distribution. This implies the previous claim.

By the previous bounds, we obtain that
$$x' \Lambda(\lambda;\sigma_+,\sigma_-)x \lesssim \frac{p}{d} - 2\sqrt{\frac{d}{s}} \int_{c_{\tau/2}+\varepsilon}^2 x w(x)\mathrm{d}x.$$
Using $p \sim \sqrt \tau d^{1/2}$, the above expression is seen to
be negative as soon as
$$s \lesssim 4p^2d \left(\frac{\int_{c_{\tau/2}+\varepsilon}^2 x w(x)\mathrm{d}x}{\tau}\right)^2 = 4(C_\tau-\varepsilon)p^2d,$$
where we put
\begin{equation}\label{eq:Ctau}
C_\tau = \left(\frac{\int_{c_{\tau/2}}^2 x
w(x)\mathrm{d}x}{\tau}\right)^2.
\end{equation}
One can easily show that the map $\tau \mapsto C_\tau$ is
increasing and that
$$C_0 = \lim_{\tau \to 0} C_\tau = 1 \qquad \text{ and } \qquad  C_1 = \left(\int_0^2 x w(x)\mathrm{d}x \right)^2 = \frac{16}{9\pi^2},$$
and the proof is complete.
\end{proof}

When $p=\min(d_1, d_2)$ is fixed and $s/d \to c$ for a constant $c
>0$ as $d\rightarrow \infty$, one can obtain the following
sharp estimate on the threshold for $\mathcal{APPT}$.

\begin{theorem}\label{thm:threshold-p-finite}
Let $\rho$ be a random induced state distributed according to
the measure $\mu_{d,s}$. Almost surely, when $d \to \infty$ and
$s\sim cd$, one has:
\begin{itemize}
\item[(i)] $\rho\in \mathcal{APPT}$, if $c > (p + \sqrt{p^2-1})^2$;
\item[(ii)] $\rho\notin \mathcal{APPT}$, if $c < (p + \sqrt{p^2-1})^2$.
\end{itemize}
\end{theorem}

To prove this result, we need the following well-known result
in random matrix theory. This result describes the behavior of the
spectrum of a Wishart matrix of parameters $(d,s)$ with
$s/d\rightarrow c$ and $d\rightarrow \infty$ \cite{bs}.

\begin{proposition}\label{prop:RMT-Wishart}
Let $\lambda_1 \geq \lambda_2 \geq \cdots \geq \lambda_d \geq 0$
be the eigenvalues of a Wishart matrix of parameters $(d, s)$.
Then, in the asymptotic regime $d \to \infty$ and $s \sim cd$ for
a constant $c>0$, one has:
\begin{itemize}
\item[(i)] Almost surely, when $d \to \infty$, the empirical eigenvalue
distribution
$$\mu(s^{-1}W) = \frac{1}{d}\sum_{i=1}^d \delta_{s^{-1}\lambda_i}$$
converges weakly to the Marchenko-Pastur (or free Poisson)
distribution $\pi_c$ given by the formula
$$\pi_c=\max (1-c,0)\delta_0+\frac{\sqrt{4c-(x-1-c)^2}}{2\pi x}1_{[(\sqrt c - 1)^2,(\sqrt c +1)^2]}(x)\,dx;$$
\item[(ii)] For every fixed integer $k$, almost surely, as $d\to \infty$
$$ \lambda_d, \lambda_{d-1}, \ldots, \lambda_{d-k+1} \to a_c=
\begin{cases}
0 &\qquad \text{ if } c \leq 1,\\
(\sqrt c - 1)^2 &\qquad \text{ if } c > 1,
\end{cases}
$$
and
$$\lambda_1, \lambda_2, \ldots, \lambda_k \to b_c=(\sqrt c +1)^2.$$
\end{itemize}
\end{proposition}

\begin{proof}[Proof of Theorem \ref{thm:threshold-p-finite}.]
Recall that a spectrum $(\lambda_1 \geq \cdots \geq \lambda_d)$
corresponds to states in $\mathcal{APPT}$ if and only if the
matrix $\Theta(\lambda;\sigma_+,
\sigma_-)=\Lambda(\lambda;\sigma_+, \sigma_-) +
\Lambda(\lambda;\sigma_+, \sigma_-)^T$ is positive for all pairs
$(\sigma_+, \sigma_-)$. Such a criterion for
APPT is invariant under scaling of the matrix $\rho$. Thus, it is
equivalent to consider the non-normalized Wishart matrix $W$.

In our case, we have that, for all $i<j$, $\sigma_+(i,j) \leq
d_1(d_1+1)/2$ and $\sigma_-(i,j) \leq d_1(d_1-1)/2$, which are
bounded quantities. Hence, it follows from Proposition
\ref{prop:RMT-Wishart} that, asymptotically, matrices
$\Theta(\lambda;\sigma_+, \sigma_-)$ are all equal to
$$\Lambda_c=
\begin{pmatrix}
2a_c & a_c-b_c &  \cdots & a_c-b_c \\
a_c-b_c & 2a_c &  \cdots & a_c-b_c \\
\vdots & \vdots & \ddots & \ \vdots \\
a_c-b_c & a_c-b_c &  \cdots & 2a_c \\
\end{pmatrix} = (a_c+b_c)\Id_{d_1} + (a_c-b_c)\mathbf{1}_{d_1}.$$
The eigenvalues of the matrix above are $(a_c+b_c) + d_1(a_c-b_c)$
with multiplicity one and $(a_c+b_c)$ with multiplicity $d_1-1$.

Note that if $c\leq 1$, then $a_c=0$ and $b_c>1$ which
implies that $(a_c+b_c) + d_1(a_c-b_c)<0$. Equivalently, the
matrix $\Lambda_c$ is not positive and $(\lambda_1, \cdots,
\lambda_d)$ does not correspond to states in $\mathcal{APPT}$. On
the other hand, if $c>1$, $\Lambda_c$ is positive if and only if
$(a_c+b_c) + d_1(a_c-b_c) \geq 0$ which can be shown to be
equivalent to the condition in the statement.
\end{proof}

Combing Theorems \ref{thm:threshold-p-infty} and \ref{thm:threshold-p-finite}, we get the following theorem; although  this statement is strictly weaker than the results above, it captures the behavior of the threshold in an unique statement.

\begin{theorem}\label{thm:threshold-p-both}
There are effectively computable absolute constants $c, C>0$, such that,  if $\rho$ is a
bipartite random induced state on $\C^d=\C^{d_1} \otimes
\C^{d_2}$, obtained by partial tracing a random pure state on
$\bC^d \otimes \C^s$, then for $d=d_1d_2$ large enough,
\begin{itemize}
\item[(i)] The random density matrix $\rho$ is {\em not} $\mathcal{APPT}$ with very large probability when $s \leq cp^2 d$;
\item[(ii)] The random density matrix $\rho$ is $\mathcal{APPT}$ with very large probability when $s \geq Cp^2 d$.
\end{itemize}
\end{theorem}

The above theorem asserts that the thresholds
for $\mathcal{APPT}$ is indeed (approximately) $4p^2d$. Together
with Theorem \ref{Threshold:APPT:existence}, one can obtain the
estimate for the mean width of $\mathcal{APPT}^\circ$.

\begin{corollary}\label{cor:mean-width}
Let $\mathcal{APPT}$ be the
set of states with APPT on the bipartite system
$\C^d=\C^{d_1}\otimes \C^{d_2}.$ Then the threshold function $s_0$ for $\mathcal{APPT}$ satisfies
$$s_0=s_0(d_1, d_2)\sim w(\mathcal{APPT}^\circ)^2\sim p^2d.$$ In particular,
$w(\mathcal{APPT}^\circ)\sim p \sqrt{d}$.
\end{corollary}

\noindent {\bf Acknowledgment.}
The research of the three authors was supported by NSERC Discovery grants
and an ERA at the University of Ottawa.
The research of BC was supported by the ANR Granma. The research of IN was supported by a PEPS grant from the Institute of Physics of the CNRS and by a travel grant APC of the University of Toulouse. IN acknowledges the hospitality of the University of Ottawa, where part of this work was done.
The research of DY has been
initiated with support from the Fields Institute, the NSERC
Discovery Accelerator Supplement Grant \#315830 from Carleton
University,
and completed while supported by a start-up grant from the
Memorial University of Newfoundland.

\vskip 2mm \noindent  Beno{i}t Collins, {\small \em Dept. of
Mathematics and Statistics, University of Ottawa, ON, Canada and CNRS, Institut Camille Jordan, Universit\'{e} Lyon 1, France.} \\
{\small\tt Email: bcollins@uottawa.ca}

\vskip 2mm \noindent  Ion Nechita, {\small \em  CNRS, Laboratoire de Physique Th{\'e}orique , IRSAMC, Universit{\'e} de Toulouse, UPS, F-31062 Toulouse, France.}\\
{\small \tt E-mail: nechita@irsamc.ups-tlse.fr}

\vskip 2mm \noindent Deping Ye, {\small \em Dept. of Mathematics
and Statistics, Memorial University of Newfoundland,
   St. John's, NL, Canada.} \\
 {\small \tt Email: deping.ye@mun.ca}

\end{document}